\documentclass{llncs}

\usepackage{amssymb}
\usepackage{latexsym}
\usepackage{graphicx}

\newcommand{\Sch}{{\bf S}}
\newcommand{\Scht}{{\bf T}}

\newcommand{\dom}{\mathsf{C}}
\newcommand{\ndom}{\mathsf{N}}

\newcommand{\vect}[1]{\bar{#1}}
\newcommand{\x}{\vect{X}}
\newcommand{\y}{\vect{Y}}
\newcommand{\z}{\vect{Y}}

\newcommand{\ra}{\rightarrow}

\newcommand{\lra}{\leftrightarrow}

\newcommand{\ftgd}[4]{\forall {#1} ( #2 \rightarrow \exists #3 ( #4) )}
\newcommand{\tgd}[2]{#1 \rightarrow #2}
\newcommand{\egd}[3]{\forall {#1} (#2 \rightarrow (#3))}

\newcommand{\pred}[1]{\mathit{#1}}
\newcommand{\attr}[1]{\mathit{#1}}
\newcommand{\const}[1]{\mathit{#1}}

\newcommand{\chase}{\mathrm{chase}}

\newcommand{\homo}{\stackrel{\mathsf{hom}}{\ra}}
\newcommand{\isom}{\stackrel{\mathsf{iso}}{\lra}}

\newcommand{\cert}[3]{\mathrm{cert}(#1,#2,#3)}

\newcounter{cefalo}

\newenvironment{plist}
  {\begin{list}{\textit{(\arabic{cefalo})}}{\usecounter{cefalo}}}
  {\end{list}}

\begin{document}

\sloppy

\title{Containment of Schema Mappings\\ for Data Exchange\\ (Preliminary Report)}

\author{Andrea Cal\`{\i}$^{1,3}$ and Riccardo
  Torlone$^{2}$}




\institute{$^1$Dept.~of Computer Science and Inf.~Systems, Birkbeck University of London, UK\\
$^2$Dip.~di Informatica e Automazione, Universit\`a Roma Tre, Italy\\
  $^3$Oxford-Man Institute of Quantitative Finance, University of
  Oxford,  UK\\[2mm]
  \texttt{andrea{@}dcs.bbk.ac.uk},\, \texttt{torlone{@}dia.uniroma3.it}
}

\maketitle

\begin{abstract}
  In data exchange, data are materialised from a source schema to a target
  schema, according to suitable source-to-target constraints.  Constraints are
  also expressed on the target schema to represent the domain of interest.  A
  schema mapping is the union of the source-to-target and of the target
  constraints.
  In this paper, we address the problem of containment of schema mappings for
  data exchange, which has been recently proposed in this framework as a step
  towards the optimization of data exchange settings.  We refer to a natural
  notion of containment that relies on the behaviour of schema mappings with
  respect to conjunctive query answering, in the presence of so-called LAV
  TGDs as target constraints.  Our contribution is a practical technique for
  testing the containment based on the existence of a homomorphism between
  special ``dummy'' instances, which can be easily built from schema mappings.
  We argue that containment of schema mappings is decidable for most practical
  cases, and we set the basis for further investigations in the topic.  This
  paper extends the preliminary results of~\cite{CaTo08}.
\end{abstract}


\section{Introduction}


In distributed database applications, restructuring information from a certain
format into a desired one is a complex task.  Consequently, there is a strong
need for methods and tools supporting the problem of transforming data coming
in different formats from multiple sources.
A clean formalisation of this problem has been recently proposed under the
name of \emph{data exchange}~\cite{FaKP05}.  In this approach, a first order
logic specification, called \emph{schema mapping}, describes declaratively how
data structured under one schema (the \emph{source schema}) are to be
transformed into data structured under a different schema (the \emph{target
  schema}): the goal is to take a given instance of the source schema and
translate it into an instance of the target schema such that the schema
mapping and the constraints (a.k.a.~dependencies) over the target are
satisfied.

For example, if we consider a source database schema with a single relation
$\pred{employee}(\attr{Name}, \attr{Salary}, \attr{Department})$ and a target
database composed by the pair of relations $\pred{person}(\attr{Name},
\attr{Address})$ and $\pred{salary}(\attr{Employee}, \attr{Amount})$, then a
typical schema mapping would be represented by the following
\emph{source-to-target dependency} $\sigma_1$:
$\tgd{\pred{employee}(X,Y,Z)}{\exists W\, \pred{person}(X,W),
  \pred{salary}(X,Y)}$.  This dependency, which is a TGD (tuple-generating
dependency), states that for each tuple $\pred{employee}(e,s,d)$ in the source
there must be two tuples in the target: $\pred{person}(e,a)$ and
$\pred{salary}(e,s)$, where $a$ is a new, unknown value.
%
A target dependency, stating that each employee name occurring in relation
$\pred{salary}$ must occur in relation $\pred{person}$ in the first position,
can also be expressed with a dependency $\sigma_2$ of the same form:
$\tgd{\pred{salary}(X,Y)}{\exists Z\, \pred{person}(X,Z)}$.
In~\cite{FKMP05} it has been shown that, for a certain class of target
dependencies, given a source instance $I$ over the source schema, a ``most
general'' solution to this problem, called universal solution,
can be computed by applying to $I$ a well-known procedure called \emph{chase}.
Intuitively, this procedure generates the solution by enforcing in $I$ the
satisfaction of the dependencies.  For example, if the instance of the above
source schema consists of the single tuple $\pred{employee}(\const{john}, 50,
\const{toys})$ then, in order to enforce the satisfaction of source-to-target
dependency, the chase generates the following tuples over the target schema:
$\pred{person}(\const{john}, v_1)$ and $\pred{salary}(\const{john},50)$, where
$v_1$ denotes a labelled null value.  The answer to a conjunctive query on a
universal instance, without considering result tuples that have at least one
null, coincides with the correct answers (called \emph{certain answers}) to
the query.

A natural question that arises in this scenario is the following: given a
schema mapping $M_1$, is there a different schema mapping $M_2$ over the same
source and target that has the same solution of $M_1$ and provides a simpler
and, possibly, more efficient way to generate such solution?  As an example,
it is possible to see that the above schema mapping is equivalent to the one
obtained by replacing $\sigma_1$ with the dependency $\sigma_3$ defined as
$\tgd{ \pred{employee}(X,Y,Z)}{\pred{salary}(X,Y)}$ since, intuitively, the
generation of a tuple in $\pred{person}$ is guaranteed, during the chase, by
the target dependency $\sigma_2$.  


More precisely, we define a \emph{schema mapping} as the union of the
source-to-target and target dependencies. 
We say that a schema mapping $M$ is \emph{contained into} another mapping $M'$
if for every source instance, the answers to every conjunctive query $q$ under
$M$ are a subset of the answers to $q$ under $M'$.  Two schema mappings are
said to be \emph{equivalent} when they are contained into each other.

Different variants of this problem have been addressed by Fagin et
al. in~\cite{FKNP08}, where three different notions of schema mapping
equivalence are proposed, and a characterisation of the problem in some
special cases is presented; however, except for one case, no positive results
are given.  However, there is still space for positive results in cases that
are interesting from the practical point of view, as we will show in the
following.

In this paper, we address the problem of containment of schema mappings for
data exchange.  We first propose a definition of schema mapping containment
which corresponds to the notion of conjunctive-query equivalence proposed
in~\cite{FKNP08}.  We provide some basic results that follow from general
properties of first-order formulae and known characterisations of the data
exchange problem.  We then consider the case of LAV TGDs, i.e., TGDs that have
a single atom in the left-hand side.  LAV TGDs are a quite general class of
TGDs that, together with some additional dependencies (which we do not
consider here) are able to represent most known formalisms for ontology
reasoning, as shown in~\cite{CaGL09} for \emph{Linear Datalog$^\pm$}, a
language whose rules are a special case of LAV TGDs, with a single atom in the
left-hand side.  Notice that, under LAV TGDs (and Linear Datalog$^\pm$ as
well), due to the existentially quantified variables in the right-hand side of
TGDs, the chase does not always terminate, forcing us to reason on instances
of unbounded size for the target schema.  Notice also that LAV TGDs generalise
the well-known class of \emph{inclusion dependencies}~\cite{AbHV95}.
We propose a practical technique for testing the containment of schema mapping
in this case: our technique relies on the construction of a finite set of
``dummy'' instances for the source schema, which depends only on the
source-to-target mappings.
The paper ends with a discussion on a number of challenging problems that
naturally follow from these preliminary steps.  To the best of our knowledge,
this is the first result of this kind in this context that follows the line of
previous approaches to the problem of checking the equivalence and
optimisation of relational expressions representing
queries~\cite{AhSU79a,AhSU79b}.  

The rest of the paper is organised as follows. In
Section~\ref{sec:preliminaries} we provide the basic definitions and recall
some useful results on the data exchange problem.  Then, in
Section~\ref{sec:containment}, we introduce the problem of schema mapping
containment and present our technique for testing the containment.  Finally,
in Section~\ref{sec:conclusions} we discuss future directions of research and
draw some conclusions.
 
\section{Preliminaries}
\label{sec:preliminaries}

\subsection{Basics}

A \emph{(relational) schema} $\Sch$ is composed by a set of \emph{relational
  predicates} $r(A_1,\ldots,A_n)$, where $r$ is the name of the predicate and
$A_1,\ldots,A_n$ are its \emph{attributes}.  A predicate having $n$ attributes
is said to be $n$-ary or, equivalently, to have \emph{arity} $n$.
An instance of a relation $r(A_1,\ldots,A_n)$ is a set of tuples, each of
which associates with each $A_i$ a value.  An instance $I$ of a schema $\Sch$
contains an instance of each relation in $\Sch$.  In the following, except
when explicitly stated, we assume that instances are finite.  We shall
consider values from two (infinite) domains: a set of constants $\dom$, and a
set of \emph{labelled nulls} $\ndom$; the latter are intended to be
placeholders for unknown constants, therefore they can be interpreted also as
existentially quantified variables.

%
As usual in data exchange, we will focus on two special kind of constraints
(a.k.a.~dependencies): tuple generating dependencies (TGDs) and equality
generating dependencies (EGDs), as it is widely accepted that they include
most of the naturally-occurring constraints on relational databases.  We will
use the symbol $\vect{X}$ to denote a sequence (or, with slight abuse of
notation, a set) of variables $X_1, \ldots, X_k$.  A TGD has the form:
$\ftgd{\x}{\phi(\x)}{\y}{\psi(\x, \z}$ where $\phi(\x)$ and $\psi(\x, \y)$ are
conjunction of atomic formulas, and allows (together with suitable EGDs) the
specification of foreign key, inclusion, and multivalued dependencies, among
others.  An EGD has the form: $\egd{\x}{\phi(\x)}{X_1=X_2}$ where $\phi(\x)$
is a conjunction of atomic formulas and $X_1$, $X_2$ are variables in $\x$: it
strictly generalises key constraints and functional
dependencies~\cite{AbHV95}, as it can be easily seen.  In both TGDs and EGDs,
the left-hand side is called \emph{body}, and the right-hand side is called
\emph{head}.  We usually omit the universal quantifiers to simplify the
notation.

A LAV TGD is a TGD having a single atom on the left-hand side, i.e., of the
form $r(\x) \ra \exists\y\, \psi(\x,\y)$.  Notice that variables can be
repeated in the body atom.

In the following, for space reasons, we will focus on TGDs only; significant
classes of EGDs can be added without changing our results; see, for instance,
the class of keys and non-conflicting TGDs in~\cite{CaGL09}.

We now provide the important notion of \emph{homomorphism}.  We consider
instances having constants and nulls 
as values.
\begin{definition}
  A \emph{homomorphism} from an instance $I$ to an instance $J$, both of the
  same schema, is a function $h$ from constant values and
  nulls 
  occurring in $I$ to constant values and nulls 
  occurring in $J$ such that: \textsl{(i)} it is the identity on constants,
  and \textsl{(ii)} $h(I) \subseteq J$, where for a tuple $t= r(c_1, \ldots,
  c_n)$ we denote $h(t) = r(h(c_1), \ldots, h(c_n))$, and for a set of tuples
  $I$, $h(I) = \{ h(t) \,\mid\, t \in I\}$.  An \emph{isomorphism} is a
  surjective and injective homomorphism.
\end{definition}

If there is a homomorphism from an instance $I$ to an instance $J$, we write
$I \homo J$.  We use the notation $I \isom J$ to denote that there is an
isomorphism from $I$ to $J$ (and therefore another one from $J$ to $I$).

\subsection{Schema mappings}

In the relational-to-relational data exchange framework~\cite{FaKP05}, a
schema mapping (also called a data exchange setting) is defined as follows.
In the following, we will assume that both souce-to-target TGDs and target
TGDs are LAV TGDs.

\begin{definition}
  A schema mapping is a 4-tuple: $M=(\Sch, \Scht, \Sigma_{st}, \Sigma_t)$,
  where:
  \begin{plist}\itemsep-\parsep
\item $\Sch$ is a source schema,
\item $\Scht$ is a target schema,
\item $\Sigma_{st}$ is a finite set of \emph{s-t} (source-to-target) LAV TGDs
  $\ftgd{\x}{\phi(\x)}{\y}{\chi(\x, \y)}$ where $\phi(\x)$ is a conjunction of
  atomic formulas over $\Sch$ and $\chi(\x, \y)$ is a conjunction of atomic
  formulas over $\Scht$, and
\item $\Sigma_t$ is a finite set of LAV target TGDs over $\Scht$.
\end{plist}
\end{definition}

Given an instance $I$ of $\Sch$, a solution for $I$ under $M$ is an instance
$J$ of $\Scht$ such that $(I, J)$ satisfies $\Sigma_{st}$ and $J$ satisfies
$\Sigma_t$.  A solution in general has values in $\dom \cup \ndom$ (nulls or
constants).
In general, there are many possible solutions, possibly an infinite number,
for $I$ under a schema mapping $M$.  A solution $J$ is \emph{universal} if
there is a homomorphism from $J$ to every other solution for $I$ under
$M$~\cite{FKMP05}.

In~\cite{FKMP05} it was shown that a universal solution of $I$ under $M$
certainly exists when the dependencies in $\Sigma_t$ are either EGDs or
\emph{weakly-acyclic} TGDs, a class of TGDs which admits limited forms of
cycles.  In this case, if a solution exists, a universal solution can be
computed 
by applying the \emph{chase procedure}~\cite{MaMS79} to $I$ using $\Sigma_{st}
\cup \Sigma_t$.  The chase\footnote{With \emph{chase}, we refer
  interchangeably to the procedure or to its resulting instance.} is a
fundamental tool that has been widely used to investigate several database
problems such as, checking implication of constraints, checking equivalence of
queries, query optimisation, and computing certain answers in data integration
settings (see, e.g.,~\cite{GrMe99}).  This procedure takes as input an
instance $D$ and generates another instance by iteratively applying
\emph{chase steps} based on the given
dependencies.  
In particular, a TGD $\ftgd{\x}{\phi(\x)}{\y}{\psi(\x, \y)}$ can be applied at
step $k$ of the chase, to the partial chase $D_{k-1}$ obtained at the previous
step, if there is a homomorphism $h$ from $\phi(\x)$ to $D_{k-1}$; in this
case, the result of its application is $D_k = D_{k-1} \cup h'(\psi(\x, \y))$,
where $h'$ is the extension of $h$ to $\y$ obtained by assigning fresh
labelled nulls to the variables in $\y$.
%
%
The chase of $I$ with respect to a set of dependencies $\Sigma$, denoted by
$\chase_\Sigma(I)$, is the instance obtained by applying all applicable chase
steps exhaustively to~$I$.  Notice that such instance may be infinite.


Universal solutions are particularly important also for query answering since
a conjunctive query\footnote{We shall henceforth consider conjunctive queries
  (a.k.a.~select-project-join queries) only, for which we refer the reader
  to~\cite{AbHV95}.} $q$ over the target scheme can be evaluated against any
universal solution.  More precisely, given a schema mapping $M=(\Sch, \Scht,
\Sigma_{st}, \Sigma_t)$, an instance $I$ of $\Sch$, and a conjunctive query
$q$ over $\Scht$, the \emph{certain answers} of $q$ on $I$ under $M$ is the
set of all tuples of constants occurring in every solution for $I$ under $M$.
Fagin et al.~\cite{FKMP05} have shown that if $q$ is a union of conjunctive
queries, the certain answers of $q$ on $I$ under $M$, denoted $\cert{q}{I}{M}$
can be obtained by evaluating $q$ over any universal solution $J$ for $I$
under $M$ and then eliminating all the tuples with nulls (in symbols:
$\cert{q}{I}{M} = q^\downarrow(J)$).



\section{Schema Mapping Containment and Equivalence}
\label{sec:containment}

\subsection{Definitions and preliminary results}

In this section we consider two schema mappings $M=(\Sch, \Scht, \Sigma_{st},
\Sigma_t)$ and $M'=(\Sch, \Scht, \Sigma'_{st}, \Sigma'_t)$.
Also, we only consider conjunctive queries. Our notion of containment refers
to the behaviour of schema mapping with respect to query answering.


\begin{definition}[Containment and equivalence of schema
  mappings~\cite{FKNP08}] \label{def:containment} We say that $M$ is contained
  in $M'$, in symbols $M\subseteq M'$, if for every instance $I$ of $\Sch$ and
  every query $q$ over $\Scht$ we have that the certain answers of $q$ on $I$
  under $M$ are contained in the certain answers of $q$ on $I$ under $M'$, in
  symbols
  \[
  \forall I \forall q\ \cert{q}{I}{M} \subseteq \cert{q}{I}{M'}
  \]
  $M$ and $M'$ are equivalent, in symbols $M \equiv M'$, if both $M\subseteq
  M'$ and $M'\subseteq M$.
\end{definition}

The following preliminary result on schema mapping containment states that:
\textsl{(i)} we can focus on universal solutions only and \textsl{(ii)} a
necessary and sufficient condition for containment relies on the existence of
a homomorphism between universal solutions.

\begin{lemma}
  \label{lem:univ} Given two schema mappings $M,M'$ as above, we have that $M
  \subseteq M'$ if and only if, for every instance $I$ of $\Sch$, there exist
  two universal solutions $J, J'$ for $I$ under $M$ and $M'$ respectively,
  such that $J \homo J'$.
\end{lemma}

\begin{proof} 
  \textit{\underline{(Only if).}}
  By contradiction, assume that $M \subseteq M'$, but there is no two
  universal solutions $J,J'$ for some instance $I$ under $M,M'$ respectively.
  Let $J = \chase_\Sigma(D)$ an $J' = \chase_{\Sigma'}(D)$.  Now, we proceed
  by induction on the number of applications of the chase step on $I$. 
  \textit{Base step.}  Take the Boolean conjunctive query obtained by
  replacing every null in $I$ with a distinct variable.  Since $M \subseteq
  M'$, not only $J \models q$ (i.e., $q$ has positive answer on $I$)
  trivially, but also by hypothesis we also have $J' \models q$.
  \textit{Inductive step.}  Assume that at the $k$-th application, the partial
  chase, denoted $\chase^{(k)}_\Sigma(D)$ maps onto $J'$ via some
  homomorphism.  Consider $\chase^{(k+1)}_\Sigma(D)$ and turn it into a query
  as above.  Similarly, we get $\chase^{(k+1)}_\Sigma(D) \homo J'$.  With this
  inductive argument we show $J \homo J'$.
  \textit{\underline{(If).}} Let $J, J'$ be universal solutions for an instance
  $I$ under $M, M'$ respectively, with $J \homo J'$.  Given a query $q$, if
  for a tuple $t$ if holds $t \in \cert{q}{I}{M}$, this amounts to say $t \in
  q^\downarrow(J)$, or equivalently there is a homomorphism from $q$ to $J$
  that sends the head variables of $q$ to constants.  Since there is also
  another homomorphism from $J$ to $J'$, we have also $t \in q^\downarrow(J)$
  and therefore $t \in \cert{q}{I}{M'}$.
%
\end{proof}

The next result easily follows from the lemma above and the results on the
generation of universal solutions for data exchange~\cite{FKMP05}.  We shall
make use of the notion of chase on the union of source-to-target and target
TGDs; in this case, we implicitly assume that the schema is the union of the
source and the target schema.

\begin{lemma}[straightforward from~\cite{FKNP08}] \label{lem:chase} Given two
  schema mappings $M,M'$ as above, we have $M\subseteq M'$ if and only for
  every instance $I$ of $\Sch$, $\chase_{\Sigma_{st}\cup\Sigma_t}(I) \homo
  \chase_{\Sigma'_{st}\cup\Sigma'_t}(I)$.
\end{lemma}

\begin{proof}
  It follows by Lemma~\ref{lem:univ} and the fact that
  $J=\chase_{\Sigma_{st}\cup\Sigma_t}(I)$ is a universal solution for $I$
  under $M$ and $J'=\chase_{\Sigma'_{st}\cup\Sigma'_t}(I)$ is a universal
  solution for $I$ under $M'$~\cite{FKMP05}.
\end{proof}

\begin{example} \label{ex:first} Let us consider the following schema mapping
  $M=(\Sch, \Scht, \Sigma_{st}, \Sigma_t)$:
\[
\begin{array}{rrcl}
  M= \{ & \Sch   & = & \{r_1(A,B)\} ,\ \ \Scht  = \{r_2(C,D),R_3(M,F)\} \\
  & \Sigma_{st} & = & \{ \tgd{r_1(X,Y)}{\exists Z\,r_2(Y,Z)}, 
  \tgd{r_1(X,Y)}{\exists Z\, r_3(Z,Y)} \} \\
  & \Sigma_t & = & \{ \tgd{r_2(X,Y)}{\exists Z\, r_3(Z,X)} \} \ \ \ \}\\
\end{array}
\]
and the generic database instance $I=\{r_1(a_i,b_i)\}$ $(1\leq i\leq n)$ of
$\Sch$.  The universal solution generated by applying the chase process to $I$
using the given dependencies is: $J = \{r_2(b_i,v_i), r_3(v'_i,b_i),
r_3(v''_i,b_i)\}$ $(1\leq i\leq n)$, where the $v_i$ are labelled nulls.
Let us now consider the following schema mapping defined over the same source
and target schemas:
\[
\begin{array}{rrcl}
  M'= \{ & \Sch   & = & \{r_1(A,B)\} ,\ \ \Scht  = \{r_2(C,D),r_3(M,F)\} \\
  & \Sigma'_{st} & = & \{ \tgd{r_1(X,Y)}{\exists Z\, r_2(Y,Z)} \} \\
  & \Sigma'_t & = & \{ \tgd{r_2(X,Y)}{\exists Z\, r_3(Z,X)} \} \ \ \ \}\\
\end{array}
\]
We have that the universal solution generated by applying the chase process to
$I$ using $\Sigma'_{st}\cup\Sigma'_t$ is: $J' = \{r_2(b_i,v_i),
r_3(v'_i,b_i)\}$ $(1\leq i\leq n)$.  It is easy to see that $J'$ and $J$ are
homomorphically equivalent, i.e., $J \homo J'$ and $J' \homo J$.  This is
actually true for every instance of $\Sch$ and so $M \equiv M'$.  Notice that
$M'$ is ``more compact'' than $M$ and so the chase using the dependencies in
$M'$ requires a lower number of steps (we do not discuss on optimality
criteria related to containment in this paper).

Finally, consider a schema mapping $M''$ as follows:
\[
\begin{array}{rrcl}
M''= \{ & \Sch   & = & \{r_1(A,B)\} ,\ \ \Scht  = \{r_2(C,D),r_3(E,F)\} \\
& \Sigma''_{st} & = & \{ \tgd{r_1(X,Y)}{\exists Z\, r_2(Y,Z)} \} \\
& \Sigma''_t & = & \{ \tgd{r_2(X,Y)}{\exists Z\, r_3(Y,X)} \} \ \ \ \}\\
\end{array}
\]
By applying the chase process to $I$ using these dependencies we obtain the
target instance: $J'' = \{r_2(b_i,v_i), r_3(v_i,b_i)\}$ $(1\leq i\leq n)$.
Now, we have $J' \homo J''$ ($v'_i\mapsto v_i$) but there is no homomorphism
from $J''$ to $J'$.  Again, this is valid in general and so $M''\subseteq M'$
but $M' \not\subseteq M''$.
\end{example}
Clearly, the test of Lemma~\ref{lem:chase} is not usable in practice since it
refers to an unbounded number of source instances, and moreover of unbounded
size in general.

\subsection{A practical test for containment}

We now present a technique for deciding containment between two schema
mappings.  We are able to show that given a schema mapping $M = (\Sch, \Scht,
\Sigma_{st}, \Sigma_t)$, in order to test containment, it suffices to test it
on a \emph{finite} set of instances for $\Sch$, depending only on $M$, that we
call \emph{dummy instances} for $M$, denoted $D_M$.
\begin{definition}
  The set of \emph{dummy databases} $D_M$ for a schema mapping $M=(\Sch,
  \Scht, \Sigma_{st}, \Sigma_t)$ is constructed as follows.  For every
  relational symbol $r$ of arity $n$ appearing in the body of some TGD in
  $\Sigma_{st}$, the corresponding instances $D_r$ are obtained by considering
  all possible tuples of the form $r(z_1, \ldots, z_n)$,
  where the $z_i$ are freshly invented constants in $\dom$, possibly with
  repetitions.
\end{definition}

\begin{example} \label{ex:second} Consider again $M=(\Sch, \Scht, \Sigma_{st},
  \Sigma_t)$ of Example~\ref{ex:first}.  Then, the dummy set $D_M$ contains
  four single-tuple instances: $\{r_1(z_1,z_2)\}$, $\{r_1(z_1,z_1)\}$,
  $\{r(z_2,z_2)\}$ and $\{r(z_2,z_1)\}$.  We do not consider $\{r(z_2,z_2)\}$
  or $\{r(z_2,z_1)\}$ (though it would not harm the decision algorithm that we
  are going to introduce in the following) because they can be obtained from
  the first two instances by null renaming.
\end{example}

We now show the main property of the dummy instances, showing that they are a
useful tool for checking containment of schema mappings.  Before stating this
paper's main theorem, we need some intermediate results.



\begin{lemma} \label{lem:chase-union} Consider two instances $D_1, D_2$
  constructed with fresh and non-fresh constants, i.e., with values in $\dom
  \cup \ndom$, and a set $\Sigma$ of LAV TGDs.  We then have
  \[
  \chase_\Sigma(D_1 \cup D_2) \isom \chase_\Sigma(D_1) \cup \chase_\Sigma(D_2)
  \]
  Without loss of generality, we assume that the two sets of nulls generated
  in $\chase_\Sigma(D_1)$ and $\chase_\Sigma(D_2)$ have empty intersection.
\end{lemma}

\begin{proof} \textit{(sketch).}  The result is proved by observing that every
  tuple $t$ in $D_1$ generates a fragment of the chase $\chase_\Sigma(\{t\})$
  which is totally independent of any other tuple in $D_1$ (the same holds, of
  course, for tuples in $D_2$).  This because the application of a LAV TGD
  depends, by definition, \emph{solely} on the single tuple to which it is
  applied, independently of the others.
\end{proof}

From the above result we immediately get the following.

\begin{lemma} \label{lem:chase-union-2} Consider two instances $D_1$ and
  $D_2$, and two sets of LAV TGDs $\Sigma_1$ and $\Sigma_2$.  We have that
  \textit{(1)} and \textit{(2)} iff \textit{(3)}, where \textit{(1), (2), (3)}
  are as follows.
   \textit{(1)} $\chase_{\Sigma_1}(D_1) \homo \chase_{\Sigma_2}(D_1)$;
   \textit{(2)} $\chase_{\Sigma_1}(D_2) \homo \chase_{\Sigma_2}(D_2)$;
   \textit{(3)} $\chase_{\Sigma_1}(D_1 \cup D_2) \homo \chase_{\Sigma_2}(D_1 \cup
    D_2)$.
\end{lemma}

We now come to our first main result.

\begin{theorem}\label{the:containment}
  Consider two schema mappings $M=(\Sch, \Scht, \Sigma_{st}, \Sigma_t)$ and
  $M'=(\Sch, \Scht, \Sigma'_{st}, \Sigma_t)$.  We have that $M \subseteq M'$
  if and only if, for every instance $D$ in the set of dummy instances $D_M$
  for $M$, we have $\chase_{\Sigma}(D) \homo \chase_{\Sigma'}(D)$, where
  $\Sigma = \Sigma_{st} \cup \Sigma_t$ and $\Sigma' = \Sigma'_{st} \cup
  \Sigma'_t$.
\end{theorem}

\begin{proof} \textit{(sketch).}\\
  \textsl{(\underline{Only if}).} It follows from
  Definition~\ref{def:containment}, since $\chase_{\Sigma}(I)$ to
  $\chase_{\Sigma'}(I)$ are both solutions for $M, M'$
  respectively (in particular, they are universal solutions).\\
  \textsl{\underline{(If)}.} Consider a generic source instance $I$.  From
  Lemma~\ref{lem:chase-union} it is easy to see that $\chase_{\Sigma}(D) \isom
  \bigcup_{t\in I}\chase_{\Sigma}\{t\}$.  Then, for each $t \in I$ that
  generates at least one tuple in the chase procedure, there is a dummy
  instance $D \in D_M$ such that $D \isom \{t\}$.  By hypothesis, for each $D
  \in D_M$ we have $\chase_{\Sigma}(D) \homo \chase_{\Sigma'}(D)$, and
  therefore for every $t \in I$ we immediately have $\chase_{\Sigma}(\{t\})
  \homo \chase_{\Sigma'}(\{t\})$.  By Lemma~\ref{lem:chase-union}, it then
  holds $\chase_{\Sigma}(D) \isom \bigcup_{t\in I}\chase_{\Sigma}\{t\} \homo
  \bigcup_{t\in I}\chase_{\Sigma'}\{t\} \isom \chase_{\Sigma'}(D)$, and by
  Lemma~\ref{lem:chase} we get $M \subseteq M'$.
\end{proof}

The above result provides an insight into the containment problem, but since
in general the chase of a dummy instance $D \in D_M$ can be infinite, it is
not obvious if checking $\chase_{\Sigma}(D) \homo \chase_{\Sigma'}(D)$ is
decidable.  The following result provides an answer to the question.

\begin{theorem} \label{the:chase-hom} Given an instance $D$ and two sets of
  LAV TGDs $\Sigma_1$ and $\Sigma_2$, such that $\chase_{\Sigma_1}(D)$ is
  finite, checking whether $\chase_{\Sigma_1}(D) \homo \chase_{\Sigma_2}(D)$
  is decidable.
\end{theorem}

\begin{proof}
  \textit{(sketch).}  The proof is a direct consequence of the results
  in~\cite{JoKl84}.  If $\chase_{\Sigma_2}$ is finite, the result is trivial.
  Let us consider the case where $\chase_{\Sigma_2}$ is infinite.
  Preliminarly, we have to introduce the notion of \emph{level} of a tuple in
  the chase of an instance $I$.  The notion of level is inductively defined as
  follows.  All tuples in $I$ have level $0$; if a tuple $t_2$ is added in a
  chase step applied on a tuple $t_1$ at level $\ell$, then $t_2$ has level
  $\ell+1$.
  To check whether $\chase_{\Sigma_1}(D) \homo \chase_{\Sigma_2}(D)$, it
  suffices to check whether there exists a homomorphism from
  $\chase_{\Sigma_1}(D)$ (which is finite) to the (finite) segment of
  $\chase_{\Sigma_2}(D)$ constitute by its first $|\Sigma_2| \cdot
  \cdot (W+1)^W$ levels, where $W$ is the maximum arity of predicates
  appearing in $\Sigma_2$.
\end{proof}

The following result is a direct consequence of Theorems~\ref{the:containment}
and~\ref{the:chase-hom}.

\begin{theorem} \label{the:final} Consider two schema mappings $M$ and $M'$ as
  above such that, for every instance $D$ in the set of dummy instance $D_M$
  for $M$, $\chase_{\Sigma_1}(D)$ is finite.  We have that checking whether $M
  \subseteq M'$ is decidable.
\end{theorem}

\begin{example} \label{ex:third} Let us consider again the schema mappings
  $M=(\Sch, \Scht, \Sigma_{st}, \Sigma_t)$ and $M'=(\Sch, \Scht, \Sigma'_{st},
  \Sigma'_t)$ of Example~\ref{ex:first}.  As shown in Example~\ref{ex:second},
  the set of dummy instances is $D_M = \{D_1, D_2\}$ where $D_1 =
  \{r_1(z_1,z_2)\}$ and $D_2 = \{r_1(z_1,z_1)\}$.  Let us now apply the chase
  procedure to $D_1$ and $D_2$ using the dependencies in $\Sigma$ and
  $\Sigma'$.  We obtain the instances: $\chase_{\Sigma}(D_1) =
  \{r_2(z_2,\alpha_1), r_3(\alpha_2,z_2), r_3(\alpha_3,z_2)\}$, and
  $\chase_{\Sigma'}(D_1) = \{r_2(z_2,\beta_1), r_3(\beta_2, z_2)\}$.  As for
  $D_2$, we immediately have $\chase_{\Sigma}(D_2) \isom \chase_{\Sigma}(D_1)$
  and $\chase_{\Sigma'}(D_2) \isom \chase_{\Sigma'}(D_1)$.  We immediately
  notice $\chase_{\Sigma}(D_1) \homo \chase_{\Sigma'}(D_1)$,
  $\chase_{\Sigma}(D_2) \homo \chase_{\Sigma'}(D_2)$, therefore $M \subseteq
  M'$.  With an analogous test we also get $M' \subseteq M$.
\end{example}

\section{Conclusions and Open Problems}
\label{sec:conclusions}

In this paper we have investigated the problem of containment and equivalence
of schema mappings for data exchange.  As a first contribution, we have
provided a technique for testing a containment $M \subseteq M'$ between two
schema mappings, that is based on the existence of a homomorphism between the
chase of ``dummy'' instances.  Our technique is applicable in the case where
both the source-to-target TGDs and the target TGDs are LAV TGDs, and the chase
of all dummy instances under the TGDs in $M$ is finite.  With respect to the
results of~\cite{FKNP08}, which is a paper closely related to our work, our
notion of containment corresponds to the notion of \emph{containment with
  respect to conjunctive query answering} in~\cite{FKNP08} and extensions thereof~\cite{CM:EDBT2010,CM:TPLP2010}.  We think this is
a natural definition if we are querying the target schema with conjunctive
queries, and it can be extended in case of different query languages (see
below).  In~\cite{FKNP08} it is also shown that the problem is undecidable for
full TGDs on the target schema (full TGDs are TGDs, not necessarily of the LAV
kind, that do not have existentially quantified variables in the head).  Our
opinion is that it is important to consider restricted classes of TGDs on the
target schema that are able to capture interesting real-world cases.  In
particular, the class of LAV TGDs that we have considered generalises
inclusion dependencies (see~\cite{AbHV95}) and also the Datalog$^\pm$
family~\cite{CaGL09}.  In particular, the latter is able to represent most
ontology formalisms that are used in the Semantic Web.  In~\cite{CaGL09} it is
also shown how relevant classes of EGDs, which generalise common database
constraints such as key and functional dependencies, can be combined with TGDs
so that TGDs and EGDs do not interact, and all the results for TGDs alone
apply also in the presence of EGDs; we refer the reader to the paper for the
details.
We note that the dependencies considered in this paper may also be understood as integrity constraints whose application may be incrementally checked \cite{M:FQAS2004,CM:LOPSTR2003,CM:AAI2000}, possibly in the presence of uncertain or diversified data~\cite{DM:FlexDBIST2006,DM:FlexDBIST2007,DM:QOIS2009,FMT:SIGMOD2012}.


We believe that this preliminary study opens a number of very challenging
issues, which will be subject of future research.
\begin{plist}\itemsep-\parsep
\item The identification of more general conditions
  under which the problem is decidable. 
  In particular, we intend to remove the condition of finiteness of the chase
  of dummy instances w.r.t.~the left-hand-side TGDs.  We conjecture that the
  containment problem is also decidable when the above condition is
  removed~\cite{GoMa09}.
\item The definition of a notion of ``minimal'' schema mapping, based on the
  efficient generation of the schema exchange solution.  This notion may rely
  on the notion of \emph{core} solution~\cite{GoNa06}; however, it will be
  important to characterise what makes a mapping better than another;
\item A method for ``reducing'' a schema mapping into a one that is minimal
  according to the above definition, and equivalent with respect to
  conjunctive query answering;
\item A definition of containment with respect to more expressive classes of
  queries than conjunctive queries; this would also require a novel definition
  of certain answers (and possibly of solutions).
\end{plist}

\paragraph{Acknowledgements.}  Andrea Cal\`\i{} was partially supported by the
EPSRC project \mbox{EP/E010865/1} \textit{Schema Mappings and Automated
  Services for Data Integration and Exchange}.

\end{document}